\documentclass{llncs}

\usepackage{amssymb, amsmath}
\usepackage{graphics}
\usepackage{graphicx}

\begin{document}

\title{Approximation Algorithms for Key Management in Secure Multicast}
\author{Agnes Chan\inst{1}, Rajmohan Rajaraman\inst{1}, Zhifeng Sun\inst{1}, and Feng Zhu\inst{2}}
\institute{Northeastern University, Boston, MA 02115, USA
\and
Cisco Systems, San Jose, CA, USA}
\maketitle
\newcounter{listCounter}
\newenvironment{LabeledProof}[1]{\noindent{\bf Proof of #1: }}{\qed}

\newcommand{\ListLengths}{\setlength{\itemsep}{0ex}\setlength{\topsep}{1ex}\setlength{\partopsep}{0ex}}

\newenvironment{NoIndentEnumerate}{\begin{list}{\arabic{listCounter}.}{
        \usecounter{listCounter}\setlength{\leftmargin}{1em}\ListLengths}}{\end{list}}

\newenvironment{mylowitemize}{\begin{list}{$\bullet$}{\setlength{\leftmargin}{1em}\ListLengths}}{\end{list}}

\newcommand{\wtNoArgs}{w}
\newcommand{\wt}[1]{w_{#1}}
\newcommand{\wtSet}[1]{W(#1)}
\newcommand{\wtSetNoArgs}{W}
\newcommand{\costNoArgs}{c}
\newcommand{\cost}[1]{c(#1)}
\newcommand{\nodeCost}[1]{\mbox{nc}(#1)}
\newcommand{\opt}[1]{\mbox{OPT}(#1)}
\newcommand{\mcast}[1]{M(#1)}

\newcommand{\junk}[1]{}

\newcommand{\Lset}{{\cal L}}
\newcommand{\LLset}{{\cal LL}}
\newcommand{\LHset}{{\cal LH}}
\newcommand{\eps}{\varepsilon}
\newcommand{\combine}[1]{\mbox{\tt combine}(#1)}
\newcommand{\combineNoArgs}{\mbox{\tt combine}}
\newcommand{\partition}[1]{\mbox{\tt partition}(#1)}

\newtheorem{Lcorol}{Corollary}

\newcommand{\combinet}[2]{\mbox{\tt combine}(#1,#2)}
\newcommand{\dist}{\Delta}
\newcommand{\ptas}[1]{\mbox{PTAS}(#1)}
\newcommand{\alg}[1]{\mbox{ALG}(#1)}
\newcommand{\optp}[1]{\mbox{OPT'}(#1)}

\begin{abstract}
Many data dissemination and publish-subscribe systems that guarantee the privacy and authenticity of the participants rely on symmetric key cryptography. An important problem in such a system is to maintain the shared group key as the group membership changes. We consider the problem of determining a key hierarchy that minimizes the average communication cost of an update, given update frequencies of the group members and an edge-weighted undirected graph that captures routing costs. We first present a polynomial-time approximation scheme for minimizing the average number of multicast messages needed for an update.  We next show that when routing costs are considered, the problem is NP-hard even when the underlying routing network is a tree network or even when every group member has the same update frequency.  Our main result is a polynomial time constant-factor approximation algorithm for the general case where the routing network is an arbitrary weighted graph and group members have nonuniform update frequencies.
\junk{
A number of data dissemination systems, such as
interactive gaming, stock data distribution, video conferencing, and
publish-subscribe systems need to guarantee the privacy and
authenticity of the participants.  Many such systems rely on symmetric
key cryptography, whereby all legitimate group members share a common
key for group communication.  An important problem in such a system is
to maintain the shared group key as the group membership changes.  The
focus of this paper is on a well-studied key management approach that
uses a hierarchy of auxiliary keys to update the shared group key and
maintain the desired security properties.  In this key hierarchy
scheme, a group controller distributes auxiliary keys to subgroups of
members; when the group membership changes, an appropriate subset of
the keys are updated and multicast to the relevant subgroups.

We consider the problem of determining a key hierarchy that minimizes
the average communication cost of an update, given update frequencies
of the group members and an edge-weighted undirected graph that
captures routing costs.  We first consider the objective of minimizing
the average number of multicast messages needed for an update (thus
ignoring the underlying routing graph), for which we present
polynomial-time approximation scheme (PTAS).  We next show that when
routing costs are considered, the problem is NP-hard even when the
underlying routing network is a tree network or even when every group
member has the same update frequency.  Our main result is a polynomial
time constant-factor approximation algorithm for the general case
where the routing network is an arbitrary weighted graph and group
members have nonuniform update frequencies.  We obtain improved
constant approximation factors for the special cases where the routing
network is a tree and when the update frequencies are uniform.
}
\end{abstract}


\section{Introduction}
\label{sec:intro}
A number of data dissemination and publish-subscribe systems, such as
interactive gaming, stock data distribution, and video conferencing,
need to guarantee the privacy and authenticity of the participants.
Many such systems rely on symmetric key cryptography, whereby all
legitimate group members share a common key, henceforth referred to as
the {\em group key}, for group communication.  An important problem in
such a system is to maintain the shared group key as the group
membership changes.  The main security requirement is {\em
confidentiality}: only valid users should have access to the multicast
data.  In particular this means that any user should have access to
the data only during the time periods that the user is a member of the
group.

There have been several proposals for multicast key distribution for
the Internet and ad hoc wireless
networks~\cite{canetti+gimnp:multicast,RFC2094,RFC2093,mittra:multicast,wong+gl:multicast}.
A simple solution proposed in early Internet RFCs is to assign each
user a {\em user key}; when there is a change in the membership, a new
group key is selected and separately unicast to each of the users
using their respective user keys~\cite{RFC2093,RFC2094}.  A major
drawback of such a key management scheme is its prohibitively high
update cost in scenarios where member updates are frequent.

The focus of this paper is on a natural key management approach that
uses a hierarchy of auxiliary keys to update the shared group key and
maintain the desired security properties.  Variations of this
approach, commonly referred to as the {\em Key Graph}\/ or the {\em
Logical Key Hierarchy}\/ scheme, were proposed by several independent
groups of
researchers~\cite{canetti+gimnp:multicast,caronni+wsp:multicast,shields+g:multicast,wallner+ha:multicast,wong+gl:multicast}.
The main idea is to have a single group key for data communication,
and have a group controller (a special server) distribute auxiliary
subgroup keys to the group members according to a key hierarchy.  The
leaves of the key hierarchy are the group members and every node of
the tree (including the leaves) has an associated {\em auxiliary}\/
key.  The key associated with the root is the shared group key.  Each
member stores auxiliary keys corresponding to all the nodes in the
path to the root in the hierarchy.  When an update occurs, say at
member $u$, then all the keys along the path from $u$ to the root are
rekeyed from the bottom up (that is, new auxiliary keys are selected
for every node on the path).  If a key at node $v$ is rekeyed, the new
key value is multicast to all the members in the subtree rooted at $v$
using the keys associated with the children of $v$ in the
hierarchy.\footnote{We emphasize here that auxiliary keys in the key
hierarchy are only used for maintaining the group key.  Data
communication within the group is conducted using the group key.} A
detailed example is given in Figure~\ref{fig:example}.  It is not hard
to see that the above key hierarchy approach, suitably implemented,
yields an exponential reduction in the number of multicast messages
needed on a member update, as compared to the scheme involving one
auxiliary key per user.

The effectiveness of a particular key hierarchy depends on several
factors including the organization of the members in the hierarchy,
the routing costs in the underlying network that connects the members
and the group controller, and the frequency with which individual
members join or leave the group.  Past research has focused on either
the security properties of the key hierarchy
scheme~\cite{canetti+mn:multicast} or concentrated on minimizing
either the total number of auxiliary keys updated or the total number
of multicast messages~\cite{snoeyink+sv:multicast}, not taking into
account the routing costs in the underlying communication network.

\subsection{Our contributions}
\label{sec:results}
In this paper, we consider the problem of designing key hierarchies
that minimize the average update cost, given an arbitrary underlying
routing network and given arbitrary update frequencies of the members,
which we refer henceforth to as weights.  Let $S$ denote the set of
all group members.  For each member $v$, we are given a weight
$\wt{v}$ representing the update probability at $v$ (e.g., a join/leave
action at $v$).  Let $G$ denote an edge-weighted undirected routing
network that connects the group members with a group controller $r$.
The cost of any multicast from $r$ to any subset of $S$ is determined
by $G$.  The cost of a given key hierarchy is then given by the
weighted average, over the members $v$, of the sum of the costs of the
multicasts performed when an update occurs at $v$.  A formal problem
definition is given in Section~\ref{sec:problem}.  

\begin{mylowitemize}
\item
We first consider the objective of minimizing the average number of
multicast messages needed for an update, which is modeled by a routing
tree where the multicast cost to every subset of the group is the
same.  For uniform multicast costs, we precisely characterize the
optimal hierarchy when all the member weights are the same, and
present a polynomial-time approximation scheme when member weights are
nonuniform.  These results appear in Section~\ref{sec:uniform}.

\item
We next show in Section~\ref{sec:hardness} that the problem is NP-hard
when multicast costs are nonuniform, even when the underlying routing
network is a tree or when the member weights are uniform.

\item
Our main result is a constant-factor approximation algorithm in the
general case of nonuniform member weights and nonuniform multicast
costs captured by an arbitrary routing graph.  We achieve a
75-approximation in general, and achieve improved constants of
approximation for tree networks (11 for nonuniform weights and 4.2 for
uniform weights).  These results are in Section~\ref{sec:nonuniform}.
\end{mylowitemize}

Our approximation algorithms are based on a simple divide-and-conquer
framework that constructs ``balanced'' binary hierarchies by
partitioning the routing graph using both the member weights and the
routing costs.  A key ingredient of our result for arbitrary routing
graphs is the algorithm of~\cite{khuller95balancing} which, given any
weighted graph, finds a spanning tree that simultaneously approximates
the shortest path tree from a given node and the minimum spanning tree
of the graph.

We have formulated the key hierarchy design as a static optimization
problem, capturing the update frequencies as weights instead of
explicitly modeling the time-varying membership of the group.  Our
formulation is applicable in scenarios where (a) the communication
group is large with frequent updates, yet the update probability of
any individual member is small; or (b) an update at a member may occur
due to reasons other than change in membership, e.g., if the key is
compromised, or if each ``member'' in the problem formulation actually
represents a collection of members in a local network, one of whom is
joining/leaving; or (c) the key hierarchy is periodically redesigned
by solving the static optimization problem.  Furthermore, the key
hierarchies that we design in this paper are simple and may be
amenable to maintain efficiently in a dynamic setting.  We plan to
investigate this aspect in future work.

\subsection{Related work}
\label{sec:related}
Variants of the key hierarchy scheme studied in this paper were
proposed by several independent
groups~\cite{canetti+gimnp:multicast,caronni+wsp:multicast,shields+g:multicast,wallner+ha:multicast,wong+gl:multicast}.
The particular model we have adopted matches the Key Graph scheme
of~\cite{wong+gl:multicast}, where they show that a balanced hierarchy
achieves an upper bound of $O(\log n)$ on the number of multicast
messages needed for any update in a group of $n$ members.
In~\cite{snoeyink+sv:multicast}, it is shown that $\Theta(\log n)$
messages are necessary for an update in the worst case, for a general
class of key distribution schemes.  Lower bounds on the amount of
communication needed under constraints on the number of keys stored at
a user are given in~\cite{canetti+mn:multicast}.
Information-theoretic bounds on the number of auxiliary keys that need
to be updated given member update frequencies are given
in~\cite{poovendran+b:multicast}.

In recent work,~\cite{lazos+p:multicast}
and~\cite{salido+lp:multicast} have studied the design of key
hierarchy schemes that take into account the underlying routing costs
and energy consumption in an ad hoc wireless network.  The results
of~\cite{lazos+p:multicast,salido+lp:multicast}, which consist of
hardness proofs, heuristics, and simulation results, are closely tied
to the wireless network model, relying on the broadcast nature of the
medium.  In this paper, we present approximation algorithms for a more
basic routing cost model given by an undirected weighted graph.

The special case of uniform multicast costs (with nonuniform member
weights) bears a strong resemblance to the Huffman encoding
problem~\cite{huffman:code}.  Indeed, it can be easily seen that an
optimal {\em binary}\/ hierarchy in this special case is given by the
Huffman code.  The truly optimal hierarchy, however, may contain
internal nodes of both degree 2 and degree 3, which contribute
different costs, respectively, to the leaves.  In this sense, the
problem seems related to Huffman coding with unequal letter
costs~\cite{karp:huffman}, for which a PTAS is given
in~\cite{golin+ky:huffman}.  The optimization problem that arises when
multicast costs and member weights are both uniform also appears as a
special case of the constrained set selection problem, formulated in
the context of website design optimization~\cite{heeringa+a:website}.
Another related problem is broadcast tree scheduling where the goal is
to determine a schedule for broadcasting a message from a source node
to all the other nodes in a heterogeneous network where different
nodes may incur different delays between consecutive message
transmissions~\cite{khuller+k:broadcast,liu:broadcast}.  Both the Key
Hierarchy Problem and the Broadcast Tree problem seek a rooted tree in
which the cost for a node may depend on the degrees of the ancestors;
however, the optimization objectives are different.

As mentioned in Section~\ref{sec:results}, our approximation algorithm
for the general key hierarchy problem uses the elegant algorithm
of~\cite{khuller95balancing} for finding spanning trees that
simultaneously approximates both the minimum spanning tree weight and
the shortest path tree weight (from a given root).  Such graph
structures, commonly referred to as {\em shallow-light trees}\/ have
been extensively studied (e.g.,
see~\cite{awerbuch+bp:network,kortsarz+p:shallow-light}).

\section{Problem definition}
\label{sec:problem}
An instance of the Key Hierarchy Problem is given by the tuple $(S,
\wtNoArgs, G, \costNoArgs)$, where $S$ is the set of group members,
$\wtNoArgs: S \rightarrow Z$ is the weight function (capturing the
update probabilities), $G = (V, E)$ is the underlying communication
network with $V \supseteq S
\cup \{r\}$ where $r$ is a distinguished node representing the
group controller, and $\costNoArgs : E \rightarrow Z$ gives the cost
of the edges in $G$.

Fix an instance $(S, \wtNoArgs, G, \costNoArgs)$.  We define a {\em
  hierarchy}\/ on a set $X \subseteq S$ to be a rooted tree $H$ whose
leaves are the elements of $X$.  For a hierarchy $T$ over $X$, the
cost of a member $x \in X$ with respect to $T$ is given by
\begin{eqnarray}
\sum_{\scriptsize \mbox{ancestor } u \mbox{ of } x} \, \, \sum_{\scriptsize
\mbox{child } v \mbox{ of } u} \mcast{T_v} \label{eqn:cost_def}
\end{eqnarray}
where $T_v$ is the set of leaves in the subtree of $T$ rooted at $v$
and for any set $Y \subseteq S$, $\mcast{Y}$ is the cost of
multicasting from the root $r$ to $Y$ in $G$.  The cost of a hierarchy
$T$ over $X$ is then simply the sum of the weighted costs of all the
members of $X$ with respect to $T$.  The goal of the Key Hierarchy
Problem is to determine a hierarchy of minimum cost.  An example
instance of the Key Hierarchy Problem, together with the calculation
of the cost of a candidate hierarchy for the instance, is given in
Figure~\ref{fig:example}.
\newcommand{\mK}{\mbox{K}}
\newcommand{\mU}{\mbox{U}}
\begin{figure}
\centering
\includegraphics[height=5cm]{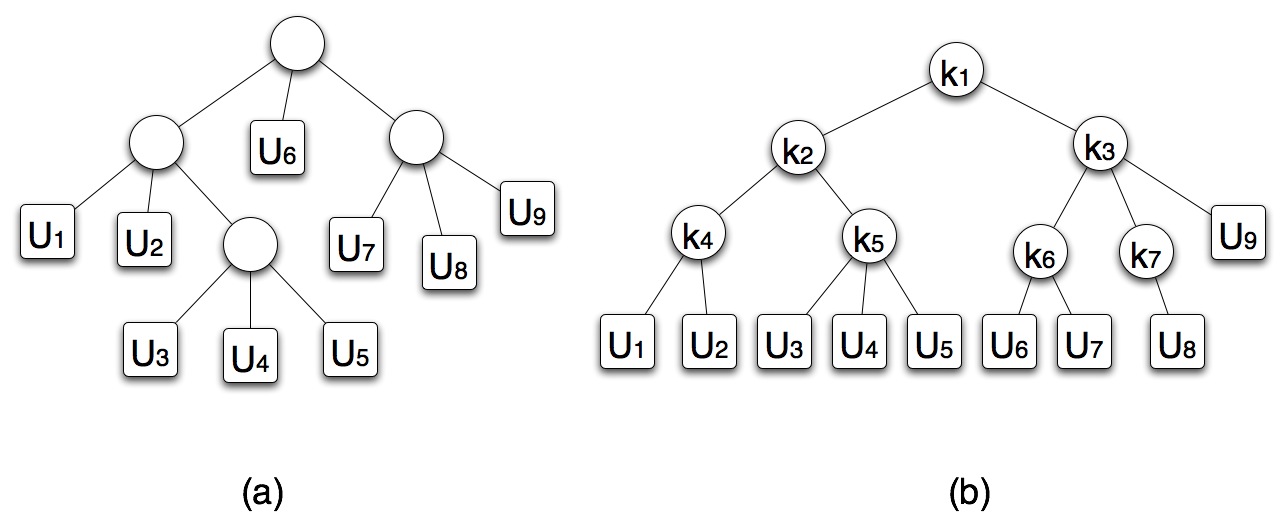}
\caption{\small An instance of the Key Hierarchy Problem with 9 group members, connected to the group
controller by a tree given in (a).  Suppose the update frequency of
every group member is 1 and the cost of every edge in the routing tree
1.  An update at member $\mU_4$ will require the rekeying of keys
$\mK_5$, $\mK_2$, and $\mK_1$.  Key $\mK_5$ is rekeyed by unicasting
to members $\mU_3$, $\mU_4$ and $\mU_5$ at a cost of 3 each.  Key
$\mK_2$ is rekeyed by multicasting to $\{\mU_1, \mU_2\}$ and to $\{\mU_3,
\mU_4, \mU_5\}$ at a cost of 3 and 5, respectively.  Finally, key
$\mK_1$ is rekeyed by multicasting to $\{\mU_1, \mU_2, \mU_3, \mU_4,
\mU_5\}$, to $\{\mU_6\}$ and to $\{\mU_7, \mU_8, \mU_9\}$ at a cost of 7,
1, and 4, respectively.  Thus, the total cost of an update at member
$\mU_4$ is 29.  Using similar calculations, the average cost of an
update can be determined to be 219/9. \label{fig:example}}
\end{figure}
We introduce some notation that is useful for the remainder of the
paper.  We use $\opt{S}$ to denote the cost of an optimal hierarchy
for $S$.  We extend the notation $\wtSetNoArgs$ to hierarchies and to
sets of members: for any hierarchy $T$ (resp., set $X$ of members),
$\wtSet{T}$ (resp., $\wtSet{X}$) denotes the sum of the weights of the
leaves of $T$ (resp., members in $X$).  Our algorithms often combine a
set ${\cal H}$ of two or three hierarchies to form another hierarchy
$T'$: $\combine{{\cal H}}$ introduces a new root node $R$, makes the
root of each hierarchy in ${\cal H}$ as a child of $R$, and returns
the hierarchy rooted at $R$.

Using the above notation, a more convenient expression for the cost of
a hierarchy $T$ over $X$ is the following reorganization of the
summation in Equation~\ref{eqn:cost_def}:
\begin{eqnarray}
\sum_{u \in T} \wtSet{T_u} \sum_{\mbox{\scriptsize child } v \, \mbox{\scriptsize of } u}
\mcast{T_v} \label{eqn:cost_def1} 
\end{eqnarray}

\section{Uniform multicast cost}
\label{sec:uniform}
In this section, we consider the special case of the Key Hierarchy
problem where the multicast cost to any subset of group members is the
same.  Thus, the objective is to minimize the average number of
multicast messages sent for an update.  We note that the number of
multicast messages sent for an update at a member $u$ is simply the
sum of the degrees of its ancestors in the hierarchy (as is evident
from Equation~\ref{eqn:cost_def}).  We start by establishing a basic
structural property of an optimal hierarchy and a lower bound on the
optimum cost.

\begin{lemma} 
\label{lem:23}
For any given member set $S$ with at least two members,
there exists an optimal hierarchy in which the degree of every
internal node is either two or three.
\end{lemma}
\begin{proof}
Let $T^*$ be an optimal hierarchy for $S$.  Since any internal node
with degree one can be replaced by its child, yielding a decrease in
cost, the degree of every internal node of $T^*$ is at least two.
Let, if possible, $v$ be an internal node of $T^*$ with degree $d \ge
4$.  We divide its children into two groups $C_1$ and $C_2$,
containing $\lceil d/2 \rceil$ and $\lfloor d/2 \rfloor$ children,
respectively.  We add two new internal nodes $v_1$ and $v_2$, make
them children of $v$, and set $v_1$ and $v_2$ to be the parents of the
nodes in $C_1$ and $C_2$, respectively.  

We now consider the cost of the new hierarchy.  The cost of any member
that does not have $v$ as an ancestor in $T^*$ does not change.  The
cost of a member that has $v$ as an ancestor in $T^*$ decreases by at
least $d - \lceil d/2 \rceil - 2 \ge 0$; thus, this cost is
nonincreasing.  If $d > 4$, there exists a member whose cost decreases
by at least $d - \lfloor d/2 \rfloor - 2 > 0$, contradicting the
optimality of $T^*$.  If $d = 4$, then we have a new hierarchy whose
cost is no more than that of $T^*$ and has fewer internal nodes with
degree greater than three.  Repeating this process until there are no
internal nodes with degree greater than 3 yields the desired claim.\qed
\end{proof}

\begin{lemma}
\label{lem:lower}
For any member set $S$, we have $\opt{S} \ge \sum_{v \in S} 3\wt{v} \log_3 (\wtSet{S}/\wt{v})$.
\end{lemma}
\begin{proof}
The proof is by induction on the size of $S$.  The claim is trivially
true for $|S| = 1$.  For the induction hypothesis, we assume that the
claim is true for member sets of size less than $m \ge 2$.  Consider
an optimal hierarchy for $S$ with $|S| = m \ge 2$.  Let the degree of
the root be $d$, and let the member set in the subtree rooted at the
$i$th child be $S_i$ with $|S_i| = m_i$, $1 \le i \le d$.  We place
the following lower bound on $\opt{S}$:
\begin{eqnarray*}
\opt{S} & \ge & d\wtSet{S} + \sum_{1 \le i \le d} \sum_{v \in S_i} 3\wt{i}
\log_3 (\wtSet{S_i}/\wt{v})\\
& = & d\wtSet{S} + 3 \wtSet{S}\sum_{1 \le i \le d} \log_3
\wtSet{S_i} - \sum_{v \in S} 3 \wt{v}
\log_3 \wt{v}\\
& \ge & d\wtSet{S} + 3 \wtSet{S} \log_3 (\wtSet{S}/d) - \sum_{v \in S} 3 \wt{v}
\log_3 \wt{v}\\
& = & d\wtSet{S} - 3\wtSet{S} \log_3 d + \sum_{v \in S} 3\wt{v} \log_3 (\wtSet{S}/\wt{v})\\
& \ge & \sum_{v \in S} 3\wt{v} \log_3 (\wtSet{S}/\wt{v}).
\end{eqnarray*}
(The third step follows from the convexity of $x \log_3
x$, the last step from $d \ge 3 \log_3 d$, $\forall d \ge 1$.)\qed
\end{proof}

\subsection{Structure of an optimal hierarchy for uniform member weights}
\label{sec:uniform.uniform}
When all the members have the same weight, we can easily characterize
an optimal key hierarchy by recursion. Let $n$ be the number of
members. When $n=1$, the key hierarchy is just a single node
tree. When $n=2$, the key hierarchy is a root with two leaves as
children. When $n=3$, the key hierarchy is a root with three leaves as
children. When $n>3$, we are going to build this key hierarchy
recursively. First divide $n$ members into 3 balanced groups, i.e. the
size of each group is between $\lfloor n/3\rfloor$ and $\lceil
n/3\rceil$. Then the key hierarchy is a root with 3 children, each of
which is the key hierarchy of one of the 3 groups built recursively by
this procedure. It is easy to verify that the cost of this hierarchy
is given by:
$$
f(n) = \left\{ \begin{array}{rl}
3n\lfloor\log_3 n\rfloor + 4(n-k) &\mbox{when $k\le n<2k$} \\
3n\lfloor\log_3 n\rfloor + 5n - 6k &\mbox{when $2k\le n<3k$}
\end{array}\right.
$$

The following theorem is due
to~\cite{heeringa:thesis,heeringa+a:website}, where this scenario arises
as a special case of the constrained set selection problem.  For
completeness, we present an alternative shorter proof here.
\begin{theorem}[\cite{heeringa:thesis,heeringa+a:website}]
\label{thm:uniform.uniform}
For uniform multicast costs and member weights, the above key
hierarchy is optimal.
\end{theorem}
\begin{proof}
We prove this by induction on the number of members. Let $n$ be the
number of members. For the base case ($n\le 5$) we can check the
optimality by brute-force. For inductive step ($n\ge 6$), we first
make two observations: optimal key hierarchies have an optimal
substructure property; and $f$ is a convex function of $n$.

By Lemma~\ref{lem:23} we know there exists an optimal hierarchy in
which the degree of the root is either two or three.  We first
consider the cse where the degree of the root is two.  Since optimal
key hierarchies satisfy the optimal substructure property, it must be
the case that the sub-hierarchies rooted at the two children of the
root must be optimal for the number of members in their respective
subtrees.  Thus, by the induction hypothesis, the cost of the optimal
hierarchy equals $f(n_1) + f(n - n_1) + 2n$, where $n_1$ is the number
of members in the subtree rooted at one of the children of the root.
Since $n \ge 6$, the convexity of $f$ implies that each subtree has at
least 3 members.  From the induction hypothesis, it also follows that
the root of each subtree has degree 3. Let the two children of the
root be $u_1$ and $u_2$.  Let the children of $u_i$ be $u_{i1}$,
$u_{i2}$, and $u_{i3}$, $1 \le i \le 2$.  We transform this hierarchy
into another key hierarchy with the same cost by adding a third child
$u_3$ to the root that has as its children $u_{13}$ and $u_{23}$. The
cost of every member in the new hierarchy remains the same as that in
the optimal hierarchy, which means this new hierarchy is also optimal
and its root has degree 3.

So we now focus on the case where there exists an optimal hierarchy in
which the root has degree 3.  Let the three children of the root have
$n_1$, $n_2$, and $n_3$ members, respectively.  It follows that the
cost of the optimal hierarchy equals $f(n_1) + f(n_2) + f(n_3) + 3n$.
The convexity of $f$ implies that the preceding cost is minimized when
each of $n_1$, $n_2$, and $n_3$ is either $\lfloor n/3 \rfloor$ or
$\lceil n/3 \rceil$.  This is precisely the proposed hierarchy, thus
completing the proof of the theorem.\qed
\end{proof}

\subsection{A polynomial-time approximation scheme for nonuniform
member weights}
\label{sec:ptas}
We give a polynomial-time approximation scheme for the Key Hierarchy
Problem when the multicast cost to every subset of the group is
identical and the members have arbitrary weights.  Given a positive
constant $\eps$, we present an polynomial-time algorithm that produces
a $(1 + O(\eps))$-approximation.  We assume that $1/\eps$ is a power
of 3; if not, we can replace $\eps$ by a smaller constant that
satisfies this condition.  We round the weight of every member up to
the nearest power of $(1 + \eps)$ at the expense of a factor of $1 +
\eps$ in approximation.  Thus, in the remainder we assume that every
weight is a power of $(1 + \eps)$.  Our algorithm $\ptas{S}$, which takes
as input a set $S$ of members with weights, is as follows.

\begin{NoIndentEnumerate}
\item
Divide $S$ into two sets, a set $H$ of the
$3^{1/\eps^2}$ members with the largest weight and the set $L = S -
H$.
\item
Initialize $\Lset$ to be the set of hierarchies consisting of one
depth-0 hierarchy for each member of $L$.
\item
Repeat the following step until it can no longer be executed: if
$T_1$, $T_2$, and $T_3$ are hierarchies in $\Lset$ with identical
weight, then replace $T_1$, $T_2$, and $T_3$ in $\Lset$ by
$\combine{\{T_1, T_2, T_3\}}$.  (Recall the definition of
$\combineNoArgs$ from Section~\ref{sec:problem}.)
\item
Repeat the following step until $\Lset$ has one hierarchy: replace the
two hierarchies $T_1$, $T_2$ with least weight by $\combine{\{T_1,
T_2\}}$.  Let $T_L$ denote the hierarchy in $\Lset$.
\item
Compute an optimal hierarchy $T^*$ for $H$.  Determine a node in $T^*$
that has weight at most $\wtSet{S} \eps$ and height at most $1/\eps$.
We note that such a node exists since every hierarchy with at least
$\ell$ leaves has a set $N$ of at least $1/\eps$ nodes at depth at
most $1/\eps$ with the property that no node in $N$ is an ancestor of
another.  Set the root of $T_L$ as the child of this node.  Return
$T^*$.
\end{NoIndentEnumerate}
We now analyze the above algorithm.  At the end of step 3, the cost of
any hierarchy $T$ in $\Lset$ is equal to $\sum_{v \in T} 3\wt{v}
\log_3 (\wtSet{T}/\wt{v})$.  If $\Lset$ is the hierarchy set at the
end of step 3, then the additional cost incurred in step 4 is at most
$\sum_{T \in \Lset} 2\wtSet{T} \log_2 (\wtSet{L}/\wtSet{T})$.

Since there are at most two hierarchies in any weight category in
$\Lset$ at the start of step 4, at least $1 - 1/\eps^2$ of the weight
in the hierarchy set is concentrated in the heaviest $4/\eps^3$
hierarchies of $\Lset$.  Step 4 is essentially the Huffman coding
algorithm and yields an optimal binary hierachy.  Using
Lemma~\ref{lem:binary} of Section~\ref{sec:nonuniform}, we note that
it achieves a $3$-approximation.  (In fact, one can show using a more
careful argument that it achieves an approximation of $2\lg ((1 +
\sqrt{5})/2)/(3\lg 3) \approx 1.52$, but the factor 3 will suffice for
our purposes here.)  This 
yields the following bound on the increase in cost due to step 4:
\[
3\left(\eps^2\wtSet{L} \log_{1 + \eps} 3 + (1 - \eps^2)\wtSet{L} \log_2
(4/\eps^2)\right) \le \wtSet{L}/\eps,
\]
for $\eps$ sufficiently small.
\junk{The cost of the hierarchy in $\LHset$ is at most
\begin{eqnarray*}
& & 3\log_3(4/\eps^2) \sum_{v \in LH} \wt{v} +  \sum_{v \in LH} 3\wt{v}
\log_3 (\wtSet{LH}/\wt{i})\\
& \le & 3\log_3(4/\eps^2) \wtSet{LH} + \opt{LH}
\end{eqnarray*}}
\junk{
We determine a node in $T_H$ that has weight at most $W\eps$ and
height at most $1/eps$.  We set the roots of $LL$ and $LH$ as children
of this node.} The final step of the algorithm increases the cost by
at most $\wtSet{L}/\eps + \eps \wtSet{S}$.  Thus, the total cost of
the final hierarchy is at most

\begin{eqnarray*}
& & \opt{H} + \opt{L} + \wtSet{L}/\eps + \wtSet{L}/\eps + \eps\wtSet{S} \\
& \le & \opt{H} + \opt{L} + 2\eps\opt{S} + \eps \opt{S}\\
& \le & (1 + 3\eps) \opt{S}.
\end{eqnarray*}  
(The second step holds since $\opt{S} \ge \sum_{v \in L} \wt{v} \log_3
(\wtSet{S}/\wt{v}) \ge \wtSet{L}/\eps^2$.)

\newcommand{\const}{C}
\section{Hardness results}
\label{sec:hardness}
In this section, we present the hardness results for Key Hierarchy
Problem with nonuniform multicast cost\junk{, i.e. the multicast cost
depends on the underlying routing structure}. First we show that the problem is strongly NP-complete if group members have different weights and the underlying routing network is a tree.  Then
we show the problem is also NP-complete if group members have the
same weights and the underlying routing network is a general graph.

\subsection{Weighted key hierarchy problem with routing tree}
Our reduction is from the NP-complete problem {\sc 3-Partition}, which
is defined as follows~\cite{garey+j:NP}.  The input consists of a set
$A$ of $3m$ elements, a bound $B\in Z^+$, and a set of sizes $S(a)\in
Z^+$ for each $a\in A$ such that $B/4<S(a)<B/2$, and $\sum_{a\in
  A}S(a)=mB$. The goal of the problem is to determine whether $A$ can
be partitioned into $m$ disjoint sets $A_1,A_2,\dots,A_m$ such that
for $1\le i\le m$, $\sum_{a\in A_i}S(a)=B$.

\begin{theorem}
\label{thm:nptree}
When group members have different weights and the routing
network is a tree, the Key Hierarchy Problem is NP-complete.
\end{theorem}
\begin{proof}
The membership in NP is immediate.  We reduce {\sc 3-partition} to the
Key Hierarchy Problem.  Let $P$ denote the given {\sc 3-Partition}
instance.  If the number $3m$ of elements in the $P$ is not a power of
three, then we add new elements in groups of three with sizes $B$,
$0$, and $0$, respectively, to make the total number of elements a
power of 3.  It is easy to verify that the original problem instance
has the desired partition if and only if the new instance has the
desired partition.  Thus, for the remainder of the proof, we assume
that the number of elements, $3m$, in $P$ is a power of $3$.

In $P$, let set $A$ be $\{a_1,a_2,\dots,a_{3m}\}$, and the size of
element $a_i$ in set $A$ be $w'_i$.  We create a routing tree $T$
consisting of a root $r$ connected to a single internal node $u$,
which in turn has edges to $3m$ leaves $v_i$ for $i=1,2,\dots,3m$, one
for each of the $3m$ members. Root $r$ is the group controller.  For
member $i$, we set its weight $w_i$ to be $w+w'_i$, where $w$ is
chosen such that $\frac{w_{max}}{w_{min}}<\frac{3\cdot 3m\log_3
  3m+1}{3\cdot 3m\log_3 3m}$, where $w_{max}=\max_i\{w_i\}$ and
$w_{min}=\min_i\{w_i\}$.  We set the cost of edge $(r,u)$ to be
$\const$, a constant which will be specified later, and the cost
of $(u,v_i)$ to be $w_i$ for $i=1,2,\dots,3m$, and the weight of leaf
$v_i$ to be $w_i$.  We now show that $P$ has a partition if and only
if the optimal key hierarchy of $T$ has cost $\const\cdot
3W\log_3 3m+W^2\cdot\left(1+1/3+1/9+\dots+1/m\right)$, where $W$ is
the sum of the weights of all the members.

If we set $\const >W^2\log_3 3m$, then the cost of an optimal key
hierarchy is smaller than $\const \cdot 3\cdot 3m\log_3 3m\cdot
w_{max}$, which is the optimal cost for $3m$ members, each with weight
$w_{max}$.  In an optimal key hierarchy, every internal node has
degree 3, since otherwise its cost is at least
$\const\cdot\left(3\cdot 3m\log_3 3m+1\right)\cdot w_{min}$, which is
not optimal given that $\frac{w_{max}}{w_{min}}<\frac{3\cdot 3m\log_3
  3m+1}{3\cdot 3m\log_3 3m}$.  So, a balanced degree-3 tree is the
only optimal key hierarchy in this case. In such a hierarchy, the cost
contributed by edge $(r,u)$ is exactly $\const\cdot 3W\log_3 3m$.  Let
$C_i$ denote the set of nodes at depth $i$ in the hierarchy, the depth
of the root being set to $0$.  By Equation~\ref{eqn:cost_def1}, the
cost contributed by edges $(u,v_i)$, $i=1,2,\dots,3m$, equals

\begin{eqnarray*}
& & \sum_{0 \le i \le \log_3 m} \sum_{x \in C_i} \wtSet{T_x} \cdot (\mcast{T_x} - \const) \\ 
& = & \sum_{0 \le i \le \log_3 m} \sum_{x \in C_i} \wtSet{T_x}^2\\
& \ge & W^2 + 3(W/3)^2 + 9(W/9)^2 + \dots + m(W/m)^2.
\end{eqnarray*}
In the last step, equality only holds when $\wtSet{T_x} = W/3^i$ for
all $x \in C_i$ (by Jensen's inequality).  Thus, the 3-partition
problem has a solution if and only if the optimal key hierarchy
achieves its minimum, which is $\const\cdot 3W\log_3
3m+W^2\cdot\left(1+1/3+1/9+\dots+1/m\right)$.\qed
\end{proof}

\subsection{Unweighted key hierarchy problem}
Our reduction is from the NP-complete {\sc 3D-Matching} problem which
is defined as follows~\cite{garey+j:NP}.  We are given finite disjoint
sets $W,U,V$ of size $q$, and a set of triples $M\subseteq W\times
U\times V$.  The goal is to determine whether there are $q$ pairwise
disjoint triples. \junk{ The proof of the following theorem is deferred to
the appendix.}

\begin{theorem}
\label{thm:npgraph}
When group members have the same key update weights and the routing
network is a general graph, the Key Hierarchy Problem is NP-complete.
\end{theorem}
\begin{proof}
We reduce {\sc 3D-Matching} to the Key Hierarchy problem.  Let $I$ be
a given instance of {\sc 3D-Matching}.  If the set size $q$ is not a
power of 3 and $q'$ is the smallest power of 3 larger than $q$, then
we construct a new instance of {\sc 3D-Matching} by adding $q' - q$
new elements to each of $W$, $U$, and $V$ as follows: for $1 \le i \le
q'-q$, add $w'_i$ to $W$, $u'_i$ to $U$, $v'_i$ to $V$, and
$(w'_i,u'_i,v'_i)$ to $M$.  It is easy to see that the original
3D-Matching instance has a solution if and only if this new
3D-Matching instance has a solution. So from now on we can assume that
$q$ is a power of 3. 

For given instance $I$, we construct a routing graph as follows.
Create vertices $w_1,w_2,\dots,w_q$ to represent each element in set
$W$, $u_1,u_2,\dots,u_q$ to represent each element in set $U$, and
$v_1,v_2,\dots,v_q$ to represent each element in set $V$. Then create
$|M|$ vertices $t_1,t_2,\dots,t_{|M|}$, and for each element
$m_i=(w_x,u_y,v_z)\in M$, add edges $(t_i,w_x),(t_i,u_y),(t_i,v_z)$ of
unit cost to the routing graph. Create another vertex $s$, and add
edges $(s,t_i)$ for $i=1,2,\dots,|M|$ of unit cost.  Finally, create
vertex $r$, and add an edge $(r,s)$ with cost $\costNoArgs$. Vertex
$r$ is the group controller, and $W \cup U \cup V$ is the set of group
members.

If we set $\costNoArgs$ to be greater than $(|M|+3q)\cdot 3\cdot
3q\log_3 3q$, then using an argument similar to the proof of
Theorem~\ref{thm:nptree}, we can show that the optimal key hierarchy
is a balanced degree-3 tree.  We will next argue that there is a
matching in $I$ if and only if the cost of the optimal key hierarchy
is $\costNoArgs\cdot 3\cdot 3q\log_3 3q + 6q(3q-1)$.

We now calculate the cost of the optimal hierarchy using
Equation~\ref{eqn:cost_def1}.  The cost contributed by edge $(r,s)$ is
exactly $\costNoArgs\cdot 3\cdot 3q\log_3 3q$.  The cost contributed
by edges $(t_i,w_x)$, $(t_i,u_y)$ and $(t_i,v_z)$ where
$i=1,2,\dots,|M|$ and $x,y,z=1,2,\dots,q$, is $\frac{9}{2}q(3q-1)$.  The cost
contributed by edges $(s,t_i)$, $i=1,2,\dots,|M|$, is at least $\frac{3}{2}q(3q-1)$.  This
minimum is achieved only if there is a 3D-Matching. So there is a
solution to the 3D-Matching problem if and only if the cost of the
optimal logical tree is $\costNoArgs\cdot 3\cdot 3q\log_3 3q +
6q(3q-1)$. And this completes the proof of the theorem.
\junk{And the cost contributed by all the edges except $(r,s)$ is at
  least $3q\left(3q+\frac{3q}{3}\right)\log_3 3q=12q^2\log_3 3q$
  (using definition~\ref{eqn:cost_def1} and processing the key
  hierarchy level by level. $3q$ is the cost contributed by is ).  }\qed
\end{proof}

\section{Approximation algorithms for nonuniform multicast costs}
\label{sec:nonuniform}
In this section, we present constant-factor approximation algorithms
for the Key Hierarchy Problem with nonuniform multicast costs.  We
first show that for any instance, there always exists a binary
hierarchy that is 3-approximate.  This guides the design of our
approximation algorithms.  We next present, in Section~\ref{sec:tree},
an 11-approximation algorithm for the case where the underlying
communication network is a tree.  Finally, we present, in
Section~\ref{sec:graph} a 75-approximation algorithm for the most
general case of our problem, where the communication network is an
arbitrary weighted graph.

\begin{lemma}
\label{lem:binary}
For any instance, there exists a 3-approximate binary hierarchy.
\end{lemma}
\begin{proof}
Consider any optimal hierarchy $T$.  Following
Equation~\ref{eqn:cost_def1}, we associate with each node $u$ of $T$ a
cost equal to $\wtSet{T_u}\sum_{\mbox{\scriptsize child }v \mbox{ of }
u}
\mcast{T_v}$; we refer to this cost as $\nodeCost{u}$.  We show how to
transform $T$ to a binary hierarchy by repeatedly replacing a node,
say $u$, with degree $d \ge 2$, by a node $u'$ of degree two and a
set $U$ of at most two other nodes, each with degree strictly less
than $d$.  To argue the bound on the cost of the binary hierarchy, we
use a charging argument: in particular, we show that $
3\nodeCost{u} \ge \nodeCost{u'} + \sum_{v \in U} 3\nodeCost{v}$.

Consider any node $u$ of $T$ of degree greater than two.  We consider
two cases.  The first case is where there is no child of $u$ that has
weight at least one-third of the weight under $u$.  We divide the
children of $u$ into two groups such that each group has at least
one-third of weight under $u$.  If such a partition exists, then we
replace $u$ by three nodes: $u'$, $u_1$, and $u_2$.  The parent of
node $u'$ is the same as the parent of $u$ (if it exists).  The node
$u'$ is the parent for both $u_1$ and $u_2$.  Finally, $u_1$ and $u_2$
are the parents of the children of $u$ in the two groups of the
partition, respectively.
\begin{eqnarray*}
3\nodeCost{u} & = & 3\wtSet{T_u} \sum_{\mbox{\scriptsize child }v \mbox{ of } u} 
\mcast{T_v}\\
& \le & \wtSet{T_u}(\mcast{T_{u_1}} + \mcast{T_{u_2}}) + 2\wtSet{T_u}
\sum_{\mbox{\scriptsize child }v \mbox{ of } u} \mcast{T_v}\\
& = & \nodeCost{u'} + 2 \wtSet{T_u} \sum_{\mbox{\scriptsize child }v
\mbox{ of } u_1} \mcast{T_v} + 
2 \wtSet{T_u} \sum_{\mbox{\scriptsize child }v
\mbox{ of } u_2} \mcast{T_v}\\
& \le & \nodeCost{u'} + 3\nodeCost{u_1} + 3 \nodeCost{u_2}.
\end{eqnarray*} 
The second case is where $u$ has a child $u_1$ with weight at least
two-third of the total weight under $u$.  In this case, we replace $u$
by two nodes $u'$ and $u_2$, with $u'$ becoming the parent of $u_1$
and $u_2$, and $u_2$ becoming the parent of the other children of $u$.
The parent of $u'$ is the same as that of $u$ (if it exists).  Using a
similar argument as above, we obtain that $3\nodeCost{u}$ equals 
$3\wtSet{T_u} \sum_{\mbox{\scriptsize child }v \mbox{ of } u}
\mcast{T_v}$, which is at most $\nodeCost{u'} + 3\nodeCost{u_2}$.\qed
\end{proof}

\subsection{Approximation algorithms for routing trees}
\label{sec:tree}
In this section, we first give an 11-approximation algorithm for the
case where weights are nonuniform and the routing network is a
tree. Then we analyze the more special case with uniform weights, and
improve the approximation factor to 4.2.

Given any routing tree, let $S$ be the set of members.  We start with
defining a procedure $\partition{\cdot}$ that takes as input the set
$S$ and returns a pair $(X,v)$ where $X$ is a subset of $S$ and $v$ is
a node in the routing tree.  First, we determine if there is an
internal node $v$ that has a subset $C$ of children such that the
total weight of the members in the subtrees of the routing tree rooted
at the nodes in $C$ is between $\wtSet{S}/3$ and $2\wtSet{S}/3$.  If
$v$ exists, then we partition $S$ into two parts $X$, which is the set
of members in the subtrees rooted at the nodes in $C$, and $S
\setminus X$.  It follows that $\wtSet{S}/3\le
\wtSet{X} \le 2\wtSet{S}/3$.  \junk{(This case is illustrated in
Figure~\ref{pic:partition}, where $Y$ denotes $S \setminus X$.)}  If
$v$ does not exist, then it is easy to see that there is a single
member with weight more than $2\wtSet{S}/3$.  In this case, we set $X$
to be the singleton set which contains this heavy node which we call
$v$.  The procedure $\partition{S}$ returns the pair $(X,v)$.  In the
remainder, we let $Y$ denote $S \setminus X$.\newline

\junk{
Having found such a partition $(X,Y)$ we recursively obtain
hierarchies for $X$ and $Y$, using the PTAS of Section~\ref{sec:ptas}
if the routing cost between the root and the partition node $v$ is a
large fraction of the total multicast cost.  Let $T_1$ and $T_2$ be
the key hierarchies thus obtained for $X$ and $Y$, respectively.  The
hierarchy for $S$ is obtained by combining $T_1$ and $T_2$ using the
$\combineNoArgs$ subroutine.  The following is our full approximation
algorithm. (PTAS is the algorithm introduced in
Section~\ref{sec:ptas}.)\newline
}

\textbf{ApproxTree($S$)}
\begin{NoIndentEnumerate}
\item If $S$ is a singleton set, then return the trivial hierarchy
with a single node.
\item $(X,v) = \partition{S}$; let $Y$ denote $S \setminus X$.
\item Let $\dist$ be the cost from root to partition node $v$. If $\dist\le \mcast{S}/5$, then let $T_1=$ApproxTree($X$); otherwise $T_1=\ptas{X}$. (PTAS is the algorithm introduced in Section~\ref{sec:ptas}.)
\item $T_2=$ApproxTree($Y$).
\item Return $\combinet{T_1}{T_2}$.
\end{NoIndentEnumerate}

\junk{
We now show that ApproxTree is a constant-factor approximation
algorithm. Let $\alg{S}$ be the key hierarchy constructed by our
algorithm, $\opt{S}$ be the optimal key hierarchy, $\optp{X}$ be the
induced key hierarchy from $\opt{S}$ by set $X$, and $\optp{Y}$ be the
induced key hierarchy from $\opt{S}$ by set $Y$. In the following
proof, we abuse our notation and use $\alg{\cdot}$ and $\opt{\cdot}$,
to refer to both the key hierarchies and their cost.  We first note
that $\optp{X}\ge \opt{X}$, $\optp{Y}\ge\opt{Y}$, and
$\opt{S}\ge\optp{X}+\optp{Y}$. So $\opt{S}\ge\opt{X}+\opt{Y}$.
}

\begin{theorem}
\label{thm:tree}
Algorithm \textbf{ApproxTree} is an $(11+\eps)$-approximation, where
$\eps > 0$ can be made arbitrarily small.
\end{theorem}
\begin{proof}
Let $\alg{S}$ be the key hierarchy constructed by our
algorithm, $\opt{S}$ be the optimal key hierarchy. In the following
proof, we abuse our notation and use $\alg{\cdot}$ and $\opt{\cdot}$
to refer to both the key hierarchies and their cost.  We first note
that $\opt{S}\ge\opt{X}+\opt{Y}$.

We prove by induction on the number of members in $S$ that
$\alg{S}\le\alpha\cdot \opt{S} + \beta\cdot \wtSet{S}\mcast{S}$, for
constants $\alpha$ and $\beta$ specified later.  The induction base,
when $|S| \le 2$, is trivial.  For the induction step, we consider
three cases depending on the distance to the partition node and
whether we obtain a balanced partition; we say that a partition $(X,
Y)$ is balanced if
$\frac{1}{3}\wtSet{S}\le\wtSet{X},\wtSet{Y}\le\frac{2}{3}\wtSet{S}$.
The first case is where $\dist\le \mcast{S}/5$ and the partition is
balanced.  In this case, we have
\begin{eqnarray*}
\alg{S} &=& \alg{X} + \alg{Y} + \wtSet{S}\left[\mcast{X}+\mcast{Y}\right] \\
 &\le& \alpha \cdot \opt{X}+\beta\cdot \wtSet{X}\mcast{X}+\alpha\cdot\opt{Y}+\beta\cdot\wtSet{Y}\mcast{Y} \\
 & & + \wtSet{S}\left[\mcast{X}+\mcast{Y}\right] \\
 &\le& \alpha\left[\opt{X}+\opt{Y}\right]+\left(\frac{2}{3}\beta+1\right)\wtSet{S}\left[\mcast{X}+\mcast{Y}\right] \\
 &\le& \alpha\cdot \opt{S}+\left(\frac{2}{3}\beta+1\right)\wtSet{S}\left[\mcast{S}+\dist\right] \\
 &\le& \alpha\cdot \opt{S}+\frac{6}{5}\left(\frac{2}{3}\beta+1\right)w(S)M(S) \\
 &\le& \alpha\cdot OPT(S)+\beta\cdot w(S)M(S)
\end{eqnarray*}
as long as $\frac{6}{5}\left(\frac{2}{3}\beta+1\right)\le \beta$,
which is true if $\beta\ge 6$.  The second case is where
$\dist>\mcast{S}/5$ and the partition is balanced. In this case, we
only call the algorithm recursively on $Y$ and use PTAS on $X$.
\begin{eqnarray*}
\alg{S} &=& \ptas{X} + \alg{Y} + \wtSet{S}\left[\mcast{X}+\mcast{Y}\right] \\
 &\le& 5(1+\varepsilon)\cdot \opt{X}+\alpha\cdot \opt{Y}+\beta\cdot \wtSet{Y}\mcast{Y} \\
 & &  +\wtSet{S}\left[\mcast{X}+\mcast{Y}\right] \\
 &\le& \alpha\left[\opt{X}+\opt{Y}\right]+\frac{2}{3}\beta \wtSet{S}\mcast{S}+2\wtSet{S}\mcast{S} \\
 &\le& \alpha\cdot \opt{S}+\left(\frac{2}{3}\beta+2\right)\wtSet{S}\mcast{S} \\
 &\le& \alpha\cdot \opt{S}+\beta\cdot \wtSet{S}\mcast{S}
\end{eqnarray*}
as long as $\alpha\ge 5(1+\varepsilon)$ and
$\frac{2}{3}\beta+2\le\beta$ which is true if $\beta\ge 6$.  The third
case is when the partition is not balanced
(i.e. $\wtSet{X}>\frac{2}{3}\wtSet{S}$). In this case, our algorithm
connects the heavy node directly to the root of the hierarchy. 

\begin{eqnarray*}
\alg{S} &=& \alg{Y} + \wtSet{S}\left[\mcast{X}+\mcast{Y}\right] \\
 &\le& \alpha\cdot \opt{Y}+\beta\cdot \wtSet{Y}\mcast{Y}+\wtSet{S}\left[\mcast{X}+\mcast{Y}\right] \\
 &\le& \alpha\cdot \opt{S}+\frac{1}{3}\beta \wtSet{S}\mcast{S}+2\wtSet{S}\mcast{S} \\
 &\le& \alpha\cdot \opt{S}+\left(\frac{1}{3}\beta+2\right)\wtSet{S}\mcast{S} \\
 &\le& \alpha\cdot \opt{S}+\beta\cdot \wtSet{S}\mcast{S}
\end{eqnarray*}
as long as $\frac{1}{3}\beta+2\le\beta$ which is true if $\beta\ge 3$. So, by induction, we have shown
$\alg{S}\le\alpha\cdot\opt{S}+\beta\cdot\wtSet{S}\mcast{S}$ for
$\alpha \ge 5(1 + \eps)$ and $\beta \ge 6$.  Since $\opt{S} \ge
\wtSet{S}\mcast{S}$, we obtain an $(11+\eps)$-approximation.\qed
\end{proof}
If the member weights are uniform, then we can improve the
approximation ratio to 4.2 using a more careful analysis of the same
algorithm.  We refer the reader to the appendix for details.

\subsection{Approximation algorithms for routing graphs}
\label{sec:graph}
In this section, we give a constant-factor approximation algorithm for
the case where weights are nonuniform and the routing network is an
arbitrary graph. In our algorithm, we compute light approximate
shortest-path trees (LAST)~\cite{khuller95balancing} of subgraphs of
the routing graph.  An $(\alpha,\beta)$-LAST of a given weighted graph
$G$ is a spanning tree $T$ of $G$ such that the the shortest path in
$T$ from a specified root to any vertex is at most $\alpha$ times the
shortest path from the root to the vertex in $G$, and the total weight
of $T$ is at most $\beta$ times the minimum spanning tree of $G$.  For
any $\gamma > 0$, the algorithm of~\cite{khuller95balancing} yields a
an $(\alpha, \beta)$-LAST with $\alpha = 1+\sqrt{2}\gamma$ and $\beta
= 1+\sqrt{2}/\gamma$, where $\gamma$ can be chosen as an input
parameter.\newline

\textbf{ApproxGraph($S$)}
\begin{NoIndentEnumerate}
\item If $S$ is a singleton set, return the trivial hierarchy with one
node.
\item Compute the complete graph on $S\cup\{root\}$. The weight of an edge $(u,v)$ is the length of shortest path between $u$ and $v$ in the original routing graph.
\item Compute the minimum spanning tree on this complete graph. Call it MST($S$).
\item 
Compute an $(\alpha, \beta)$-LAST $L$ of MST($S$).
\item $(X,v) = \partition{L}$.
\item Let $\dist$ be the cost from root to partition node $L$. If
$\dist\le\mcast{S}/5$, then let $T_1=$ApproxGraph($X$). Otherwise, $T_1=\ptas{X}$.
\item $T_2=$ApproxGraph($Y$).
\item Return $\combinet{T_1}{T_2}$.
\end{NoIndentEnumerate}

The optimum multicast to a member set is obtained by a minimum Steiner
tree, computing which is NP-hard.  It is well known that the minimum
Steiner tree is 2-approximated by a minimum spanning tree (MST) in the
metric space connecting the root to the desired members (the metric
being the shortest path cost in the routing graph).  So at the cost of
a factor 2 in the approximation, we define $\mcast{S}$ to be the cost
of the MST connecting the root to $S$ in the complete graph $G(S)$
whose vertex set is $S\cup\{root\}$ and the weight of edge $(u,v)$ is
the shortest path distance between $u$ and $v$ in the routing graph.

\begin{theorem}
The algorithm \textbf{ApproxGraph} is a constant-factor approximation.
\end{theorem}
\begin{proof}
We prove by induction on the number of members in $S$ that
$\alg{S}\le\alpha\cdot \opt{S} + \beta\cdot \wtSet{S}\mcast{S}$, for
constants $\alpha$ and $\beta$ specified later.  The induction base,
when $|S| \le 2$, is trivial.  For the induction step, we consider
three cases. The first case is $\dist\le \mcast{S}/5$ and the
partition is balanced (as defined in the proof of
Theorem~\ref{thm:tree}). Let $M_L(S)$ be the multicast cost to $S$ in
LAST. From the description of LAST we know $M_L(S)\le
\left(1+\sqrt{2}/\gamma\right)\cdot \mcast{S}$. Also we have
$M_L(S)\ge M_L(X)+M_L(Y)-\dist \ge
\mcast{X}+\mcast{Y}-\dist$. So $\left(1+\sqrt{2}/\gamma\right)\cdot
\mcast{S}\ge \mcast{X}+\mcast{Y}-\dist$.
\begin{eqnarray*}
\alg{S} &=& \alg{X} + \alg{Y} + \wtSet{S}\left[\mcast{X}+\mcast{Y}\right] \\
 &\le& \alpha \cdot \opt{X}+\beta\cdot \wtSet{X}\mcast{X}+\alpha\cdot \opt{Y}+\beta\cdot \wtSet{Y}\mcast{Y} \\
 & & + \wtSet{S}\left[\mcast{X}+\mcast{Y}\right] \\
 &\le& \alpha\left[\opt{X}+\opt{Y}\right]+\left(\frac{2}{3}\beta+1\right)\wtSet{S}\left[\mcast{X}+\mcast{Y}\right] \\
 &\le& \alpha\cdot \opt{S}+\left(\frac{2}{3}\beta+1\right)\wtSet{S}\left[\left(1+\sqrt{2}/\gamma\right)\mcast{S}+\dist\right] \\
 &\le& \alpha\cdot \opt{S}+\left(\frac{6}{5}+\sqrt{2}/\gamma\right)\left(\frac{2}{3}\beta+1\right)\wtSet{S}\mcast{S} \\
 &\le& \alpha\cdot OPT(S)+\beta\cdot \wtSet{S}\mcast{S}
\end{eqnarray*}
as long as $\left(\frac{6}{5}+\sqrt{2}/\gamma\right)\left(\frac{2}{3}\beta+1\right)\le\beta$.

The second case is $\dist>\mcast{S}/5$ and the partition is
balanced. In this case, we only call the algorithm recursively on $Y$
and use the PTAS for $X$. Since $\dist>\mcast{S}/5$, the distance from
the root to any element in $X$ is at least
$\frac{\dist}{1+\sqrt{2}\gamma}=\frac{\mcast{S}}{5(1+\sqrt{2}\gamma)}$. So
the multicast cost to any subset of $X$ is between
$\frac{\mcast{S}}{5(1+\sqrt{2}\gamma)}$ and $\mcast{S}$. By using the
PTAS, we have a $5(1 + \eps)(1+\sqrt{2}\gamma)$-approximation on
$\opt{X}$. So we have the following bound on $\alg{S}$.
\begin{eqnarray*}
\alg{S} &=& \ptas{X} + \alg{Y} + \wtSet{S}\left[\mcast{X}+\mcast{Y}\right] \\
 &\le& 5\left(1+\sqrt{2}\gamma\right)(1+\eps) \cdot \opt{X}+\alpha\cdot \opt{Y}+\beta\cdot \wtSet{Y}\mcast{Y} \\
 & & +\wtSet{S}\left[\mcast{X}+\mcast{Y}\right] \\
 &\le& \alpha\left[\opt{X}+\opt{Y}\right]+\frac{2}{3}\beta \wtSet{S}\mcast{S}+2\wtSet{S}\mcast{S} \\
 &\le& \alpha\cdot \opt{S}+\left(\frac{2}{3}\beta+2\right)\wtSet{S}\mcast{S} \\
 &\le& \alpha\cdot \opt{S}+\beta\cdot \wtSet{S}\mcast{S}
\end{eqnarray*}
as long as $\alpha\ge 5\left(1+\sqrt{2}\gamma\right)(1+\eps)$ and $\beta\ge 6$.

The third case is when the partition is not balanced. In this case,
our algorithm connect the heavy node directly to the root of key
hierarchy. So we have the following bound on $\alg{S}$.
\begin{eqnarray*}
\alg{S} &=& \alg{Y} + \wtSet{S}\left[\mcast{X}+\mcast{Y}\right] \\
 &\le& \alpha\cdot \opt{Y}+\beta\cdot \wtSet{Y}\mcast{Y}+\wtSet{S}\left[\mcast{X}+\mcast{Y}\right] \\
 &\le& \alpha\cdot \opt{S}+\frac{1}{3}\beta \wtSet{S}\mcast{S}+2\wtSet{S}\mcast{S} \\
 &\le& \alpha\cdot \opt{S}+\left(\frac{1}{3}\beta+2\right)\wtSet{S}\mcast{S} \\
 &\le& \alpha\cdot \opt{S}+\beta\cdot \wtSet{S}\mcast{S}
\end{eqnarray*}
as long as $\beta\ge 3$. So, this algorithm has a constant approximation.

So, by induction, we have shown
$\alg{S}\le\alpha\cdot\opt{S}+\beta\cdot\wtSet{S}\mcast{S}$, implying
an $(\alpha+\beta)$-approximation. When $\gamma=7$, from the
constraints, we obtain $\alpha\ge 54$ and $\beta\ge 21$. So we have a
$75$-approximation.\qed
\end{proof}

\section{Discussion}
\label{sec:disc}
We have presented a constant-factor approximation algorithm for the
Key Hierarchy Problem for the general case where the member weights
are nonuniform and the communication network is an arbitrary graph.
While we do obtain improved approximation factors when the
communication network is a tree, the factors achieved are large and
need to be improved.  We have also given a polynomial-time
approximation scheme for the problem instance where all multicasts
cost the same.  We do not know, however, whether this problem is
NP-complete.  As discussed in Section~\ref{sec:related}, the problem
is related to the classic Huffman coding problem with nonuniform
letter costs, whose complexity (P vs NP-hardness) is also not yet
resolved.

There are several other directions for future research.  We are
currently exploring the dynamic maintenance of our key hierarchies,
explicitly modeling the joining and leaving of members, while
maintaining the constant-factor approximation in cost.  We would also
like to study the design of key hierarchies where the members have a
bound on the number of auxiliary keys they store.  Also of interest is
the case where we have no (or limited) information on the update
frequencies of the members.



\appendix
\junk{
\section{Optimal hierarchy for uniform multicast cost and uniform weights}
\begin{LabeledProof}{Theorem~\ref{thm:uniform.uniform}}
We prove this by induction on the number of members. Let $n$ be the
number of members. For the base case ($n\le 5$) we can check the
optimality by brute-force. For inductive step ($n\ge 6$), we first
make two observations: optimal key hierarchies have an optimal
substructure property; and $f$ is a convex function of $n$.

By lemma~\ref{lem:23} we know there exists an optimal hierarchy in
which the degree of the root is either two or three.  We first
consider the cse where the degree of the root is two.  Since optimal
key hierarchies satisfy the optimal substructure property, it must be
the case that the sub-hierarchies rooted at the two children of the
root must be optimal for the number of members in their respective
subtrees.  Thus, by the induction hypothesis, the cost of the optimal
hierarchy equals $f(n_1) + f(n - n_1) + 2n$, where $n_1$ is the number
of members in the subtree rooted at one of the children of the root.
Since $n \ge 6$, the convexity of $f$ implies that each subtree has at
least 3 members.  From the induction hypothesis, it also follows that
the root of each subtree has degree 3. Let the two children of the
root be $u_1$ and $u_2$.  Let the children of $u_i$ be $u_{i1}$,
$u_{i2}$, and $u_{i3}$, $1 \le i \le 2$.  We transform this hierarchy
into another key hierarchy with the same cost by adding a third child
$u_3$ to the root that has as its children $u_{13}$ and $u_{23}$. The
cost of every member in the new hierarchy remains the same as that in
the optimal hierarchy, which means this new hierarchy is also optimal
and its root has degree 3.

So we now focus on the case where there exists an optimal hierarchy in
which the root has degree 3.  Let the three children of the root have
$n_1$, $n_2$, and $n_3$ members, respectively.  It follows that the
cost of the optimal hierarchy equals $f(n_1) + f(n_2) + f(n_3) + 3n$.
The convexity of $f$ implies that the preceding cost is minimized when
each of $n_1$, $n_2$, and $n_3$ is either $\lfloor n/3 \rfloor$ or
$\lceil n/3 \rceil$.  This is precisely the proposed hierarchy, thus
completing the proof of the theorem.
\end{LabeledProof}
}

\junk{
\section{NP-completeness proof}
\begin{LabeledProof}{Theorem~\ref{thm:npgraph}}
We reduce {\sc 3D-Matching} to the Key Hierarchy problem.  Let $I$ be
a given instance of {\sc 3D-Matching}.  If the set size $q$ is not a
power of 3 and $q'$ is the smallest power of 3 larger than $q$, then
we construct a new instance of {\sc 3D-Matching} by adding $q' - q$
new elements to each of $W$, $U$, and $V$ as follows: for $1 \le i \le
q'-q$, add $w'_i$ to $W$, $u'_i$ to $U$, $v'_i$ to $V$, and
$(w'_i,u'_i,v'_i)$ to $M$.  It is easy to see that the original
3D-Matching instance has a solution if and only if this new
3D-Matching instance has a solution. So from now on we can assume that
$q$ is a power of 3. 

For given instance $I$, we construct a routing graph as follows.
Create vertices $w_1,w_2,\dots,w_q$ to represent each element in set
$W$, $u_1,u_2,\dots,u_q$ to represent each element in set $U$, and
$v_1,v_2,\dots,v_q$ to represent each element in set $V$. Then create
$|M|$ vertices $t_1,t_2,\dots,t_{|M|}$, and for each element
$m_i=(w_x,u_y,v_z)\in M$, add edges $(t_i,w_x),(t_i,u_y),(t_i,v_z)$ of
unit cost to the routing graph. Create another vertex $s$, and add
edges $(s,t_i)$ for $i=1,2,\dots,|M|$ of unit cost.  Finally, create
vertex $r$, and add an edge $(r,s)$ with cost $\costNoArgs$. Vertex
$r$ is the group controller, and $W \cup U \cup V$ is the set of group
members.

If we set $\costNoArgs$ to be greater than $(|M|+3q)\cdot 3\cdot
3q\log_3 3q$, then using an argument similar to the proof of
Theorem~\ref{thm:nptree}, we can show that the optimal key hierarchy
is a balanced degree-3 tree.  We will next argue that there is a
matching in $I$ if and only if the cost of the optimal key hierarchy
is $\costNoArgs\cdot 3\cdot 3q\log_3 3q + 6q(3q-1)$.

We now calculate the cost of the optimal hierarchy using
Equation~\ref{eqn:cost_def1}.  The cost contributed by edge $(r,s)$ is
exactly $\costNoArgs\cdot 3\cdot 3q\log_3 3q$.  The cost contributed
by edges $(t_i,w_x)$, $(t_i,u_y)$ and $(t_i,v_z)$ where
$i=1,2,\dots,|M|$ and $x,y,z=1,2,\dots,q$, is $\frac{9}{2}q(3q-1)$.  The cost
contributed by edges $(s,t_i)$, $i=1,2,\dots,|M|$, is at least $\frac{3}{2}q(3q-1)$.  This
minimum is achieved only if there is a 3D-Matching. So there is a
solution to the 3D-Matching problem if and only if the cost of the
optimal logical tree is $\costNoArgs\cdot 3\cdot 3q\log_3 3q +
6q(3q-1)$. And this completes the proof of the theorem.
\junk{And the cost contributed by all the edges except $(r,s)$ is at
  least $3q\left(3q+\frac{3q}{3}\right)\log_3 3q=12q^2\log_3 3q$
  (using definition~\ref{eqn:cost_def1} and processing the key
  hierarchy level by level. $3q$ is the cost contributed by is ).  }
\end{LabeledProof}
}

\section{Improved approximation for the routing tree case when weights
are uniform}

For the case where group members have the same key update probability
and the communication network is a tree, using the same algorithm we
can show the approximation ratio is 4.2 by a different analysis, shown
as follows.

\begin{claim}
Balanced partition node can always be found if the members have the same key update weight.
\end{claim}
\begin{proof}
Suppose this kind of partition node doesn't exist, which means for all internal node $v$, its number of leaves is either $< n/3$ or $> 2n/3$. We call nodes with less than $n/3$ leaves small nodes, and nodes with more than $2n/3$ leaves large nodes. Consider the large node with only small nodes as its children, there must be a combination of its children whose total number of leaves is between $n/3$ and $2n/3$. This means this kind of partition node exists.\qed
\end{proof}

\begin{lemma}
$\alg{S} \le \alg{X} + \alg{Y} + \frac{4\dist}{3\log 3/2}n\log n + n\left(\mcast{X}+\mcast{Y}\right)$.
\end{lemma}
\begin{proof}
(1) The cost of nodes in $\alg{Y}$ is the same as their cost in $\alg{S}$. (2) Similarly, the cost of nodes in $\alg{X}$ is equal to their cost in $\alg{S}+\frac{4\dist}{3\log 3/2}n\log n$. The reason we add $\frac{4\dist}{3\log 3/2}n\log n$ is the multicast cost of each node in $\alg{X}$ increased by $\dist$ compared to its cost in $\alg{X}$. Since in the worst case $\alg{X}$ has $\log_{\frac{3}{2}} \frac{n}{3}$ levels, the increased cost is at most $2|X|\dist\log_{\frac{3}{2}} \frac{n}{3} \le 2\frac{2n}{3}\dist\log_{\frac{3}{2}} \frac{n}{3} \le \frac{4\dist}{3\log 3/2}n\log n$. Combine (1) and (2), then add the cost of the root of $\alg{X}$ and $\alg{Y}$, we know this lemma is correct.\qed
\end{proof}

\begin{lemma}
$\opt{S}\ge \opt{X}+\opt{Y}+\frac{3\dist}{\log 3}n\log n$
\end{lemma}
\begin{proof}
To any subset of $X$, the multicast cost calculated in $\opt{S}$ is $\dist$ more than the cost calculated in $\opt{X}$. From Theorem~\ref{thm:uniform.uniform}, we know the increased cost is at least $3\dist\cdot n\log_3 n=\frac{3\dist}{\log 3}n\log n$.\qed
\end{proof}

\begin{theorem}
This is a 4.2-approximation algorithm.
\end{theorem}
\begin{proof}
\begin{eqnarray*}
 & & \alg{S} \\
 &\le& \alg{X} + \alg{Y} + \frac{4\dist}{3\log 3/2}n\log n + n\left(\mcast{X}+\mcast{Y}\right) \\
 &\le& \alpha\cdot \opt{X} +\beta\cdot |X|\mcast{X} + \alpha\cdot \opt{Y} + \beta\cdot |Y|\mcast{Y} + \frac{4\dist}{3\log 3/2}n\log n \\
 & & + n\left(\mcast{X}+\mcast{Y}\right) \\
 &\le& \alpha\left[\opt{X}+\opt{Y}\right] + \frac{4\dist}{3\log 3/2}n\log n + \left(\frac{2}{3}\beta + 1\right)n\left(\mcast{X}+\mcast{Y}\right)\\
 &\le& \alpha\left[\opt{S}-\frac{3\dist}{\log 3}n\log n \right]+ \frac{4\dist}{3\log 3/2}n\log n+ \left(\frac{2}{3}\beta + 1\right)n\left(\mcast{X}+\mcast{Y}\right)  \\
 &\le& \alpha\left[\opt{S}-\frac{3\dist}{\log 3}n\log n \right]+ \dfrac{4\dist}{3\log 3/2}n\log n + \left(\frac{2}{3}\beta + 1\right)n\left(\mcast{S}+\dist\right)
 \end{eqnarray*}

\begin{eqnarray*}
 &=& \alpha\cdot \opt{S} + \left(\frac{2}{3}\beta + 1\right)n\mcast{S} - \alpha\cdot \frac{3\dist}{\log 3}n\log n + \frac{4\dist}{3\log 3/2}n\log n\\
 &\le& \alpha\cdot \opt{S} + \beta\cdot n\mcast{S}
\end{eqnarray*}

as long as $\alpha\ge 1.2$ and $\beta\ge 3$. This means this is a 4.2-approximation.\qed
\end{proof}

\end{document}